\documentclass[a4paper,USenglish,cleveref,autoref,thm-restate]{socg-lipics-v2021}
\usepackage{mathtools}
\usepackage{cite}
\newcommand{\ang}[1]{\langle #1\rangle}

\renewcommand{\ang}[1]{\langle #1\rangle}

\newcommand{\RE}{\mathbb{R}}            
\newcommand{\UU}{\mathbb{U}}
\newcommand{\XX}{\mathbb{X}}
\newcommand{\eps}{\varepsilon}          
\newcommand{\ST}{\,:\,}                 
\newcommand{\bd}{\partial}

\DeclareMathOperator{\interior}{int}
\DeclareMathOperator{\Vor}{Vor}

\usepackage{xcolor}
\newcommand{\auguste}[1]{}
\newcommand{\dave}[1]{}

\title{Voronoi Diagrams in the Hilbert Metric}
\author{Auguste H. Gezalyan}{Department of Computer Science, University of Maryland, College Park, USA \and \url{}}{octavo@umd.edu}{ADD ORCID LATER}{}
\author{David M. Mount}{Department of Computer Science, University of Maryland, College Park, USA \and \url{https://www.cs.umd.edu/~mount/}}{mount@umd.edu}{https://orcid.org/0000-0002-3290-8932}{}
\authorrunning{A.\,H.\,Gezalyan and D.\,M.\,Mount}

\Copyright{Auguste H. Gezalyan and David M. Mount}


\date{\today}

\EventEditors{John Q. Open and Joan R. Access}
\EventNoEds{2}
\EventLongTitle{42nd Conference on Very Important Topics (CVIT 2016)}
\EventShortTitle{CVIT 2016}
\EventAcronym{CVIT}
\EventYear{2016}
\EventDate{December 24--27, 2016}
\EventLocation{Little Whinging, United Kingdom}
\EventLogo{}
\SeriesVolume{42}
\ArticleNo{23}

\begin{document}

\maketitle

\begin{abstract}

The Hilbert metric is a distance function defined for points lying within a convex body. It generalizes the Cayley-Klein model of hyperbolic geometry to any convex set, and it has numerous applications in the analysis and processing of convex bodies. In this paper, we study the geometric and combinatorial properties of the Voronoi diagram of a set of point sites under the Hilbert metric. Given any convex polygon $K$ bounded by $m$ sides, we present two algorithms (one randomized and one deterministic) for computing the Voronoi diagram of an $n$-element point set in the Hilbert metric induced by $K$. Our randomized algorithm runs in $O(m n + n (\log n)(\log m n))$ expected time, and our deterministic algorithm runs in time $O(m n \log n)$. Both algorithms use $O(m n)$ space. We show that the worst-case combinatorial complexity of the Voronoi diagram is $\Theta(m n)$.
\end{abstract}

\section{Introduction}

The Hilbert metric was introduced by David Hilbert in 1895~\cite{hilbert1895linie}. Given a convex body $K$ in $d$-dimensional space, it defines a distance function between any pair of points in the interior of $K$. (Definitions are presented in Section~\ref{sec:funk-hilbert}.) The Hilbert geometry has a number of natural properties. For example, geodesics in the Hilbert geometry are straight line segments. 

It generalizes hyperbolic geometry by adapting the Cayley-Klein model of hyperbolic geometry (on Euclidean balls) to any convex body. It is also invariant under projective transformations. Hilbert geometry provides new insights into classical questions from convexity theory. It also provides new insights into the study of metric and differential geometries (such as Finsler geometries). An excellent resource on the Hilbert geometries is the handbook on Hilbert geometry by Papadopoulos and Troyanov~\cite{papadopoulos2014handbook}.

We came to consider the Hilbert geometry because of its relevance to the topic of convex approximation. Efficient approximations of convex bodies have been applied to a wide range of applications, including approximate nearest neighbor searching both in Euclidean space~\cite{AFM17a} and more general metrics~\cite{AAFM19}, optimal construction of $\eps$-kernels~\cite{AFM17b}, solving the closest vector problem approximately~\cite{EHN11,RoV21,EiV21,NaV19}, computing approximating polytopes with low combinatorial complexity~\cite{AFM17c,AAFM20}. These works all share one thing in common---they approximate a convex body by covering it with elements that behave much like metric balls. These covering elements go under various names: Macbeath regions, Macbeath ellipsoids, Dikin ellipsoids, and $(2,\eps)$-covers. While these all behave like metric balls, the question is in what metric space? Abdelkader and Mount showed that these shapes are, up to constant factors, equivalent to Hilbert balls~\cite{AbM18}. Thus, a deeper understanding of how to compute within the Hilbert geometry can lead in a principled way to a deeper understanding of convex approximation.

In spite of its obvious appeals, there has been remarkably little work the design of algorithms in the Hilbert geometry on convex polygons and polytopes. One notable exception is the work of Nielsen and Shao, which investigates properties and efficient construction of Hilbert balls in convex polygons~\cite{nielsen2017balls}.

In this paper, we investigate perhaps the most fundamental computational question one might ask about a metric geometry: How to construct the Voronoi diagram of a set of point sites in the plane? Given any convex polygon $K$ bounded by $m$ sides, we present two algorithms for computing the Voronoi diagram of an $n$-element point set in the Hilbert metric induced by $K$. The first is a randomized incremental algorithm, which runs in $O(m n + n (\log n)(\log m n))$ expected time. The second is a deterministic algorithm based on divide-and-conquer, which runs in time $O(m n \log n)$. Both algorithms use $O(m n)$ space. We show that the worst-case combinatorial complexity of the Voronoi diagram is $\Theta(m n)$, so both algorithms are worst-case optimal up to logarithmic factors.

\section{Preliminaries}

Throughout, a \emph{convex body} $K$ in $\RE^d$ is a closed, compact, full-dimensional convex set in $\RE^d$. Let $\bd K$ and $\interior(K)$ denote its boundary and interior, respectively. Given points $p, q \in \RE^d$, let $\|p - q\|$ denote the Euclidean distance between these points. Given two distinct points $p, q \in K$, let $\chi(p,q)$ denote the \emph{chord} defined as the intersection of the line passing through $p$ and $q$ with $K$.

\subsection{Funk and Hilbert Metrics} \label{sec:funk-hilbert}

Recall that a \emph{metric space} is a pair $(\XX, d)$ where $\XX$ is a set and $d$ is a function $d: \XX \times \XX \rightarrow \RE^{\geq 0}$, which, for all $x, y, z \in \XX$ satisfies
\begin{itemize}
    \item $d(x,y) = 0 ~\Leftrightarrow~ x = y$
    \item $d(x,y) = d(y,x)$
    \item $d(x,y) \leq d(x,z) + d(z,y)$
\end{itemize}

Before defining the Hilbert metric, it is convenient to define a simpler (asymmetric) distance function called the \emph{Funk weak metric}. 

\begin{definition}[Funk weak metric]
Given a convex body $K$ in $\RE^d$ and two distinct points $p, q \in \interior(K)$, let $y$ denote point where a ray shot from $p$ to $q$ intersects $\bd K$. Define the \emph{Funk weak metric} to be 
\[
    F_K(p,q)
        ~ = ~ \ln \frac{\|p - y\|}{\|q - y\|},
\]
and define $F_K(p,p) = 0$ (see Figure~\ref{fig:hilbert-funk}(a)).
\end{definition}

\begin{figure}[htbp]
    \centerline{\includegraphics[scale=0.40]{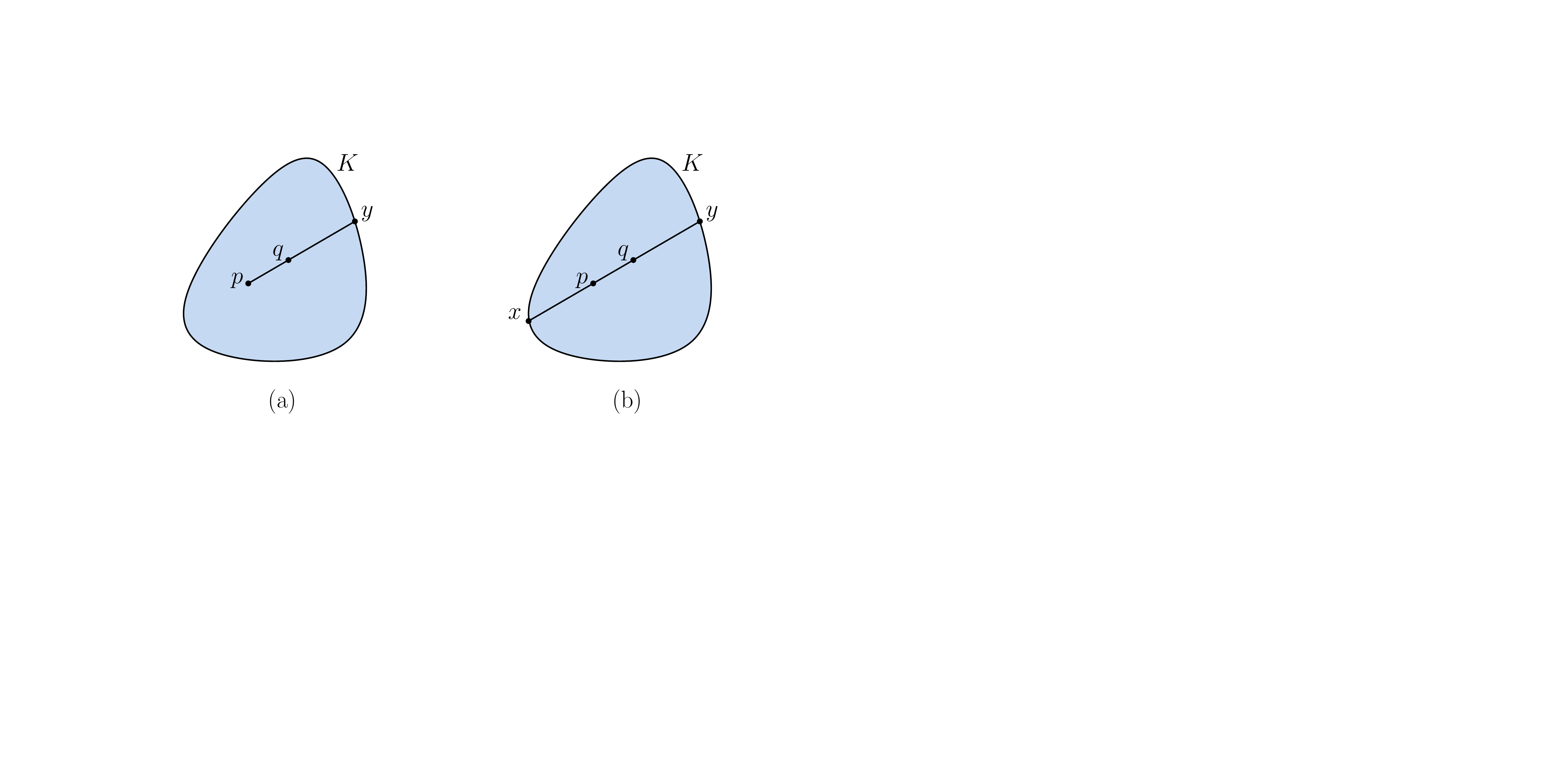}}
    \caption{(a) The Funk weak metric and (b) the Hilbert metric.}\label{fig:hilbert-funk}
\end{figure}

Observe that in the limit as $q$ approaches $y$ along the chord $\chi(p,y)$ the Funk distance increases to $+\infty$, and thus, the boundary is infinitely far away from any interior point of $K$. The Funk weak metric is not symmetric, but it satisfies the other two properties of a metric. (For a proof of the triangle inequality, see Yamada~\cite{yamada2014convex}.) Symmetrizing this yields the Hilbert metric~\cite{hilbert1895linie}.

\begin{definition}[Hilbert metric] 
Given a convex body $K$ in $\RE^d$ and two distinct points $p, q \in \interior(K)$, let $x$ and $y$ denote endpoints of the chord $\chi(p,q)$, so that the points are in the order $\ang{x, p, q, y}$. Define the \emph{Hilbert metric} to be
\[
    H_K(p,q)
        ~ = ~ \frac{F_K(p,q) + F_K(q,p)}{2}
        ~ = ~ \frac{1}{2} \ln \frac{\|p - y\|\|q - x\|}{\|q - y\|\|p - x\|},
\]
and define $H_K(p,p) = 0$ (see Figure~\ref{fig:hilbert-funk}(b)).
\end{definition}

This is indeed a metric since it satisfies the triangle inequality. Further,  if $q$ lies on the line segment between $p$ and $r$, then $H_K(p,q) + H_K(q,r) = H_K(p,r)$. As in the Funk weak metric, the boundary of $K$ is infinitely far away from any interior point. Observe that the quantity in the $\ln$ term in the definition of Hilbert is the cross ratio of $(p, q; y, x)$. It follows that the Hilbert metric is invariant under projective transformations. It is well known that straight line segments are geodesics in the Hilbert metric, but generally there may be multiple shortest paths between two points, and hence geodesics need not be line segments (see, e.g., \cite{busemann1955geodesics}). For further information, see the first chapters of the handbook on Hilbert geometry by Papadopoulos and Troyanov~\cite{papadopoulos2014handbook}.

\subsection{Balls in the Hilbert Metric}
Given a convex body $K$ in $\RE^d$, $p \in \interior(K)$, and $r \geq 0$, let 
\[
    B_K(p,r)
        ~ = ~ \{ q \in K \ST H_K(p, q) \leq r \}
\]
denote the \emph{Hilbert ball} of radius $r$ centered at $p$. Nielsen and Shao~\cite{nielsen2017balls} provided a characterization of Hilbert balls when $K$ is an $m$-sided convex polygon in $\RE^2$, showing that the ball is a polygon bounded by at most $2 m$ sides. We generalize their result to arbitrary convex polytopes in $\RE^d$. This provides an alternative (and more elementary) proof that Hilbert balls are convex (also proved in~\cite{papadopoulos2014funk,troyanov2014funk}).

\begin{lemma} \label{lem:hilbert-ball}
Given any convex polytope $K$ in $\RE^d$ bounded by $m$ facets, $p \in \interior(K)$, and $r \geq 0$, $B_K(p, r)$ is a convex polytope bounded by at most $m(m-1)$ facets. If $d = 2$, then $B_K(p, r)$ is a convex polygon bounded by at most $2 m$ sides.
\end{lemma}

\begin{proof}
For now, let us take $K$ to be any convex body (not necessarily a polytope). Let $\UU^d$ denote the set of all unit vectors in $\RE^d$. For each $u \in \UU^d$, consider the ray emanating from $p$ in direction $u$, and define $q_r(u) \in K$ to be the (unique) point along this ray whose Hilbert distance from $p$ is $r$. Clearly, the boundary of $B_K(p,r)$ is just the set of points $q_r(u)$ over all $u \in \UU^d$. 

For any $u \in \UU^d$, let $h(u)$ denote any supporting hyperplane of $K$ at the point where the ray emanating from $p$ in the direction $u$ intersects the boundary of $K$ (see Figure~\ref{fig:hilbert-ball}(a)). Define $h(-u)$ analogously for the ray emanating from $p$ in the direction $-u$. Let $g(u)$ ($=g(-u)$) denote the $(d-2)$-dimensional affine subspace where $h(u)$ and $h(-u)$ intersect. (To avoid dealing with objects at infinity, let us assume that $h(u)$ and $h(-u)$ are not parallel, which we can achieve through an infinitesimal perturbation of $K$.) Let $h_r(u)$ denote the $(d-1)$-hyperplane passing through both $g(u)$ and $q_r(u)$, and define $h_r(-u)$ analogously for $q_r(-u)$. (These are both well defined because neither $q_r(u)$ nor $q_r(-u)$ can lie on $g(u)$.) 

\begin{figure}[htbp]
    \centerline{\includegraphics[scale=0.35]{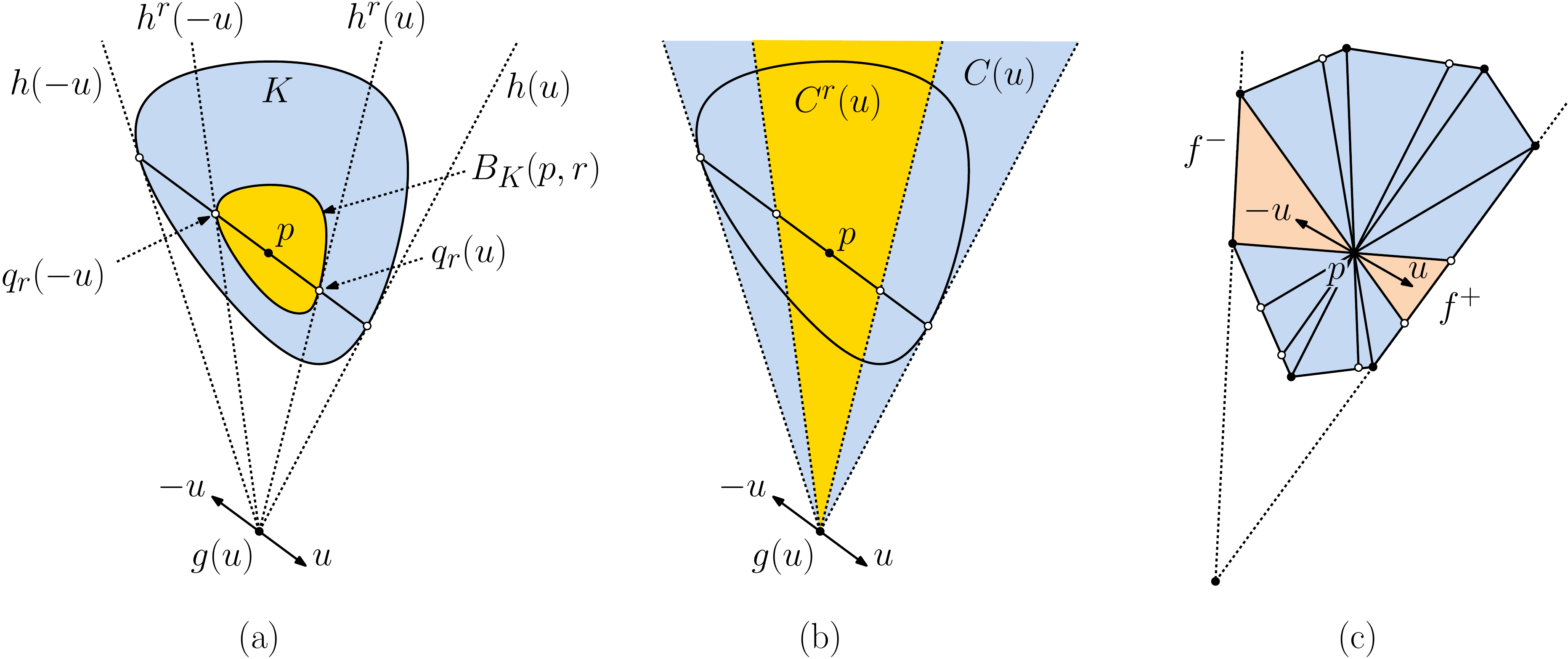}}
    \caption{Proof of Lemma~\ref{lem:hilbert-ball} (showing just the intersection with the plane $f$).}\label{fig:hilbert-ball}
\end{figure}

To simplify matters, let us restrict attention to the $2$-dimensional affine subspace $f$ that is orthogonal to $g(u)$ and passes through $p$. It follows that $g(u) \cap f$ is a single point, which we may take to be the origin. The respective intersections of $h(u)$, $h(-u)$, $h_r(u)$, and $h_r(-u)$ with $f$ are all lines that pass through the origin (since they all contain $g(u)$). Let $C(u)$ denote the convex cone on $f$ whose apex is at the origin and is bounded by $h(u)$ and $h(-u)$ (see Figure~\ref{fig:hilbert-ball}(b)). It follows from basic properties of projective geometry that every line passing through $p$ on this plane is cut by $h(u)$, $h(-u)$, $h_r(u)$, and $p$ into four points that share the same cross ratio. This implies that every point on $h_r(u)$ is at Hilbert distance $r$ from $p$. The same applies symmetrically to $h_r(-u)$, and therefore the sub-cone of $C(u)$ bounded by $h_r(u)$ and $h_r(-u)$ is just the Hilbert ball of radius $r$ centered at $p$ for $C(u)$. (This was observed by Nielsen and Shao in their analysis of the two-dimensional case.) Let us call this sub-cone $C_r(u)$.

We can now extend this observation back to the full $d$-dimensional space, by considering the Minkowski sum $f \oplus g(u)$ (that is, the set resulting from pairwise sums of points of $f$ and points of $g(u)$). Let $\widehat{C}(u) = C(u) \oplus g(u)$ and let $\widehat{C}_r(u) = C_r(u) \oplus g(u)$. We assert that $\widehat{C}_r(u)$ is the Hilbert ball of radius $r$ centered at $p$ with respect to $\widehat{C}(u)$. This follows because ratios of lengths are preserved under orthogonal projection, and so any chord of $\widehat{C}(u)$ that passes through $p$ is cut by $\widehat{C}_r(u)$ is the same proportions as is the orthogonal projection of the chord onto $f$, which we have shown above yields Hilbert distances equal to $r$. 

Note that $p \in K \subseteq \widehat{C}_r(u)$, and therefore Hilbert distances from $p$ in $\widehat{C}_r(u)$ are at least as large as they are in $K$. Therefore, for all $u \in \UU^d$, $B_K(p, r) \subseteq \widehat{C}_r(u)$, and hence
\[
    B_K(p, r)
        ~ \subseteq ~ \bigcap\nolimits_{u \in \UU^d} \widehat{C}_r(u).
\]
On the other hand, by definition of $\widehat{C}_r(u)$, we know that $\widehat{C}_r(u)$ and $B_K(p, r)$ both cover the exactly the same portion of the chord parallel to $u$ passing through $p$. Since these chords cover all the boundary points of $B_K(p,r)$, we have
\[
    B_K(p, r)
        ~ = ~ \bigcap\nolimits_{u \in \UU^d} \widehat{C}_r(u).
\]
Clearly, $\widehat{C}_r(u)$ is convex for any $u \in \UU^d$, and because convex sets are closed under (infinite) intersections, $B_K(p, r)$ is also convex, as desired.

Suppose now that $K$ is a convex polytope in $\RE^d$ bounded by $m$ $(d-1)$-dimensional facets. We say that two facets $f^+$ and $f^-$ are \emph{complimentary} if there exists $u \in \UU^d$ such that the rays emanating from $p$ in the directions $u$ and $-u$ hit $f+$ and $f^-$, respectively (see Figure~\ref{fig:hilbert-ball}(c)). We can partition the elements of $\UU^d$ into equivalence classes according to the associated complimentary pair. All the unit vectors $u$ from any one equivalence class share the same supporting hyperplanes $h(u)$ and $h(-u)$, and therefore all of them contribute the same two facets the boundary of $B_K(p,r)$. Clearly, there are at most $\binom{m}{2}$ complimentary facet pairs, and hence there are at most $2 \binom{m}{2} = m(m-1)$ facets bounding $B_K(p,r)$, as desired.

When $d = 2$ and $K$ is an $m$-sided convex polygon, the number of complimentary facets (or complimentary edges) is at most $2 m$. To see this, consider the chords $\chi(v,p)$ for each each of the $m$ vertices of $K$. This partitions $\bd K$ into at most $2 m$ boundary intervals, each corresponding to a different complimentary pair. This completes the proof.
\end{proof}

\section{Characterization of Voronoi Diagrams in the Hilbert Metric}

Using our understanding of Hilbert balls we can characterize Voronoi diagrams in the Hilbert Metric. Throughout, let $K$ denote a convex polygon in $\RE^2$, and unless otherwise stated, distances will be in the Hilbert metric induced by $K$, which we denote simply by $d(\cdot,\cdot)$. Let $S$ denote a set of $n$ points lying within $K$'s interior, which we call \emph{sites}. 

For $p \in S$, define its \emph{Voronoi cell} to be
\[
    V(p)
        ~ = ~ V_S(p)
        ~ = ~ \big\{ q \in K \ST d(q,p) \leq d(q,p'), \,\forall p' \in S \setminus \{p\} \big\}.
\]
Although points on $K$'s boundary are infinitely far from points in $K$'s interior, we can compare the relative distances a fixed boundary point and two interior points by considering the limit as an interior point approaches this boundary point. The \emph{Voronoi diagram} of $S$ in the Hilbert metric induced by $K$, denoted $\Vor_K(S)$, is the cell complex of $K$ induced by the Voronoi cells $V(p)$ for all $p \in S$. We assume that the points of $S$ are in general position, and in particular, the line passing through any pair of sites of $S$ and the lines extending any two edges of $K$ are not coincident at a common point (including all three being parallel). When this happens, the bisectors separating Voronoi cells can widen into 2-dimensional regions.

Recall that $R \subseteq \RE^d$ is a \emph{star} (or is \emph{star-shaped}) with respect to a point $p \in R$ if for each $q \in R$, the line segment $p q$ lies within $R$. We next show that Hilbert Voronoi cells satisfy this.

\begin{lemma} \label{lem:star-shaped}
Voronoi cells in the Hilbert Metric are stars with respect to their defining sites.
\end{lemma}

\begin{proof}
If this were not the case, there would exist a site $p$ and points $x,y \in K$ such that $x \in V(p)$, $y \notin V(p)$, and $y$ lies on the line segment $p x$. By collinearity, $d(p,x) = d(p,y) + d(x,y)$. Letting $q$ be the closest site to $y$, we have $d(q,y) < d(p,y)$ and $d(p,x) \leq d(q,x)$. Combining these we have $d(q,x) + d(p,y) > d(q,y) + d(p,x)$, or equivalently $d(q,x) > d(q,y) + d(x,y)$. But this violates the triangle inequality, yielding a contradiction. 

\end{proof}

\subsection{Bisectors in the Hilbert Metric} \label{sec:bisector}

Given two sites $p, p' \in \interior(K)$, we define their Hilbert bisector, denoted the \emph{$(p,p')$-bisector}, to be $\{z \in K : d_K(z,p) = d_K(z,p')\}$. We will explore the conditions for a point $z$ to lie on the bisector. Let $x$ and $y$ denote the endpoints of the chord $\chi(z,p)$ and define $x'$ and $y'$ analogously for $\chi(z,p')$. Label these points in the order $\ang{x, z, p, y}$ and $\ang{x', z, p', y'}$ (see Figure~\ref{fig:boundary-bisector}(a)). Finally, let $\ell_x$, $\ell_p$, and $\ell_y$ denote the lines passing through the line segments $x x'$, $p p'$ and $y y'$, respectively. If these three lines are coincident on some point $q$ then by basic properties of projective geometry, the cross ratios $(z,p; y,x)$ and $(z,p'; y',x')$ are equal. It follows that $d_K(z,p) = d_K(z,p')$, and hence $z$ is on the bisector. If not, then the cross ratios are different and the Hilbert distances are different. Observe that as $z$ approaches the boundary of $K$ (on the same side of $\ell_p$ as $\ell_x$), the points $z$, $x$, and $x'$ converge on a common point on $\bd K$, and $\ell_x$ approaches a support line for $K$ at this point (see Figure~\ref{fig:boundary-bisector}(b)). Thus, we obtain the following characterization of bisector points.

\begin{figure}[htbp]
    \centerline{\includegraphics[scale=0.35]{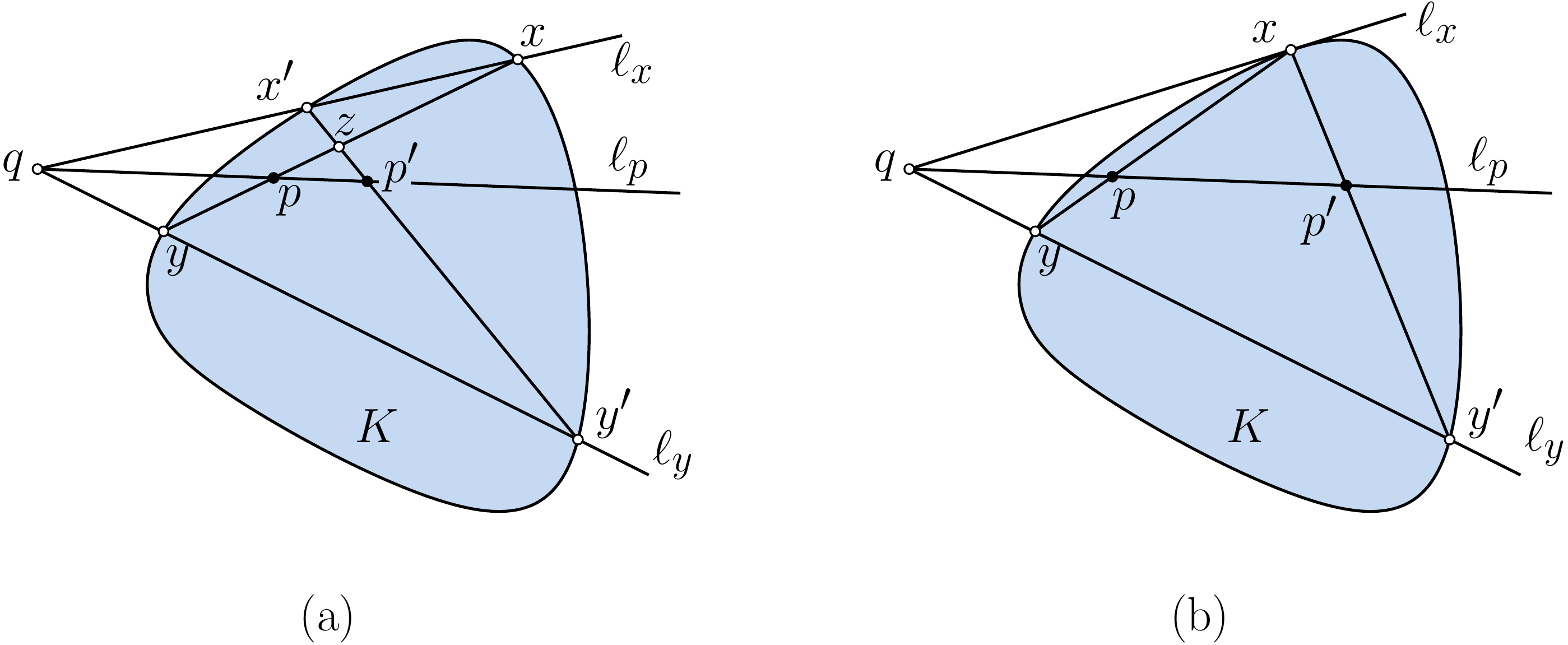}}
    \caption{Conditions for a point $z$ to lie on the Hilbert bisector between sites $p$ and $p'$.} \label{fig:boundary-bisector}
\end{figure}

\begin{lemma} \label{lem:boundary-bisector}
Given a convex body $K$ in $\RE^2$, sites $p, p' \in \interior(K)$ and any other point $z \in \interior(K)$, $z$ lies on the $(p,p')$-bisector if any only if lines $\ell_x$, $\ell_p$, and $\ell_y$ (defined above) are coincident. Further, a point $x \in \bd K$ is on the Hilbert bisector if the coincidence holds when $\ell_x$ is any support line at $x$.
\end{lemma}

When $K$ is an $m$-sided convex polygon in $\RE^2$, we can provide a more precise characterization of the bisectors. Given two sites $p, p' \in \interior(K)$, we will show below that the $(p,p')$-bisector is a piecewise curve where the pieces can be described as rational parametric functions of bounded degree. Assuming this for now, we can characterize the breakpoints in this curve. Letting $\{v_1, \ldots, v_m\}$ denote $K$'s vertices, the $2 m$ chords $\chi(v_i,p)$ subdivide $K$ into $2 m$ triangular regions, which we call $p$'s \emph{sectors} with respect to $K$ (see Figure~\ref{fig:sector}(a)%
\footnote{Throughout, figures showing Hilbert bisectors are illustrative of general concepts and are not geometrically accurate.}
).

\begin{figure}[htbp]
    \centerline{\includegraphics[scale=0.35]{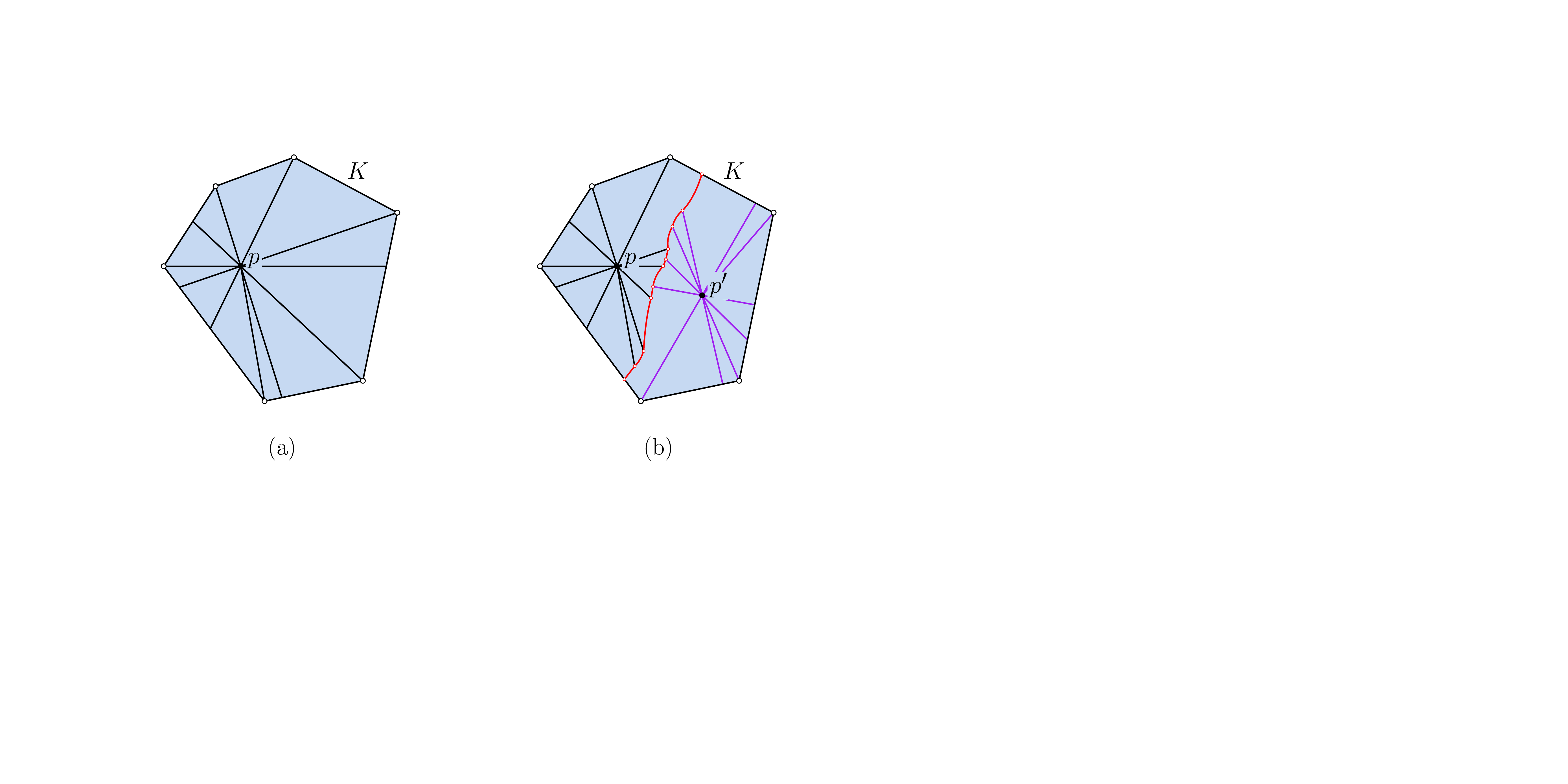}}
    \caption{(a) The sectors of $p$ with respect to $K$ and (b) a (not geometrically accurate) rendering of the Hilbert bisector between $p$ and $p'$.} \label{fig:sector}
\end{figure}

\subsection{Bisector Segments and Combinatorial Complexity} \label{sec:complexity}

For each point $z$ on the $(p,p')$-bisector, let $e$ and $f$ denote the edges where the endpoints of chord $\chi(z,p)$ intersect $\bd K$, and let $e'$ and $f'$ denote the corresponding edges of $\chi(z,p')$. Observe that the pair $(e, f)$ is uniquely determined by the sector of $p$ containing $z$, and the pair $(e', f')$ is similarly determined by the sector of $p'$ containing $z$. Therefore, the points of the $(p,p')$-bisector can be grouped into a discrete set of equivalence classes based on their sector memberships with respect to $p$ and $p'$ (see Figure~\ref{fig:sector}(b)). We refer to these as \emph{bisector segments}. Observe that by the star-shaped nature of Voronoi cells, as we travel along the bisector, we encounter the (up to $2 m$) sectors of $p$ in cyclic order (say clockwise as in the figure), and we visit (up to $2 m$) the sectors of $p'$ in the opposite cylic order (counterclockwise). We will see below that, each bisector segment can be described by a simple parametric function involving $p$, $p'$, and the four edges defining the equivalence class. Star-shapedness implies that the bisector is simply connected. Combining the total number of sectors involved for each site, we have:

\begin{lemma} \label{lem:bisector-complexity}
Given an $m$-sided convex polygon $K$ in $\RE^2$ and sites $p, p' \in \interior(K)$, the $(p,p')$-bisector is a simply connected piecewise curve consisting of at most $4 m$ bisector segments.
\end{lemma}

Given a set of $n$ sites $S$ in $K$, the Voronoi diagram consists of a collection of $n$ Voronoi cells. The intersection of two cells $V(p)$ and $V(p')$ (if nonempty) is a portion of $(p,p')$-bisector, which call a \emph{Voronoi edge}. As shown in the above lemma, each such edge is composed of at most $4 m = O(m)$ bisector segments. A cell's boundary may also contain a portion of the boundary of $K$. The intersection of two Voronoi edges (if nonempty) is a \emph{Voronoi vertex}. Because the diagram is a planar graph, we have the following bounds by a straightforward application of Euler's formula. 

\begin{lemma} \label{lem:total-complexity}
Given an $m$-sided convex polygon $K$ in $\RE^2$ and a set of $n$ sites $S$ in $K$, $\Vor_K(S)$ has $n$ Voronoi cells, at most $3 n$ Voronoi edges, and at most $2 n$ Voronoi vertices. Each Voronoi edge consists of at most $4 m$ bisector segments. Therefore, the entire diagram has total combinatorial complexity $O(m n)$. The average number of Voronoi edges per cell is $O(1)$, and the average number of bisector segments per cell is $O(m)$.
\end{lemma}

The following lemma shows that the bound on the combinatorial complexity is tight in the worst case.      

\begin{lemma} \label{lem:lower-bound}
There exists a convex polygon $K$ with $m$ sides and a set $S$ of $n$ sites within $K$ such that $\Vor_K(S)$ has combinatorial complexity $\Omega(m n)$.
\end{lemma}

\begin{proof}
We start the construction with an axis parallel rectangle $R$ that is slightly taller than wide, and the set $S$ consists of $n$ points positioned on a short horizontal line segment near the center of this rectangle (see Figure~\ref{fig:lower-bound}(a)). (This violates our general position assumptions, but the construction works if we perturb the sides of the rectangle.) Also, create a set of $m-4$ points $V = \{v_1, \ldots, v_{m-4}\}$ along $R$'s left edge. We can adjust the spacing of points of $V$ and $S$ and the side lengths of $R$ so that there exists a diamond shape (shaded in orange in the figure) so that for any point $q$ in this region, the chord $\chi(v_i,q)$ intersects the boundary of $R$ along its vertical sides. Observe that within the diamond, each consecutive pair of sites contributes an edge to the Voronoi diagram, which is easily verified to be a vertical line segment within the diamond.

\begin{figure}[htbp]
    \centerline{\includegraphics[scale=0.35]{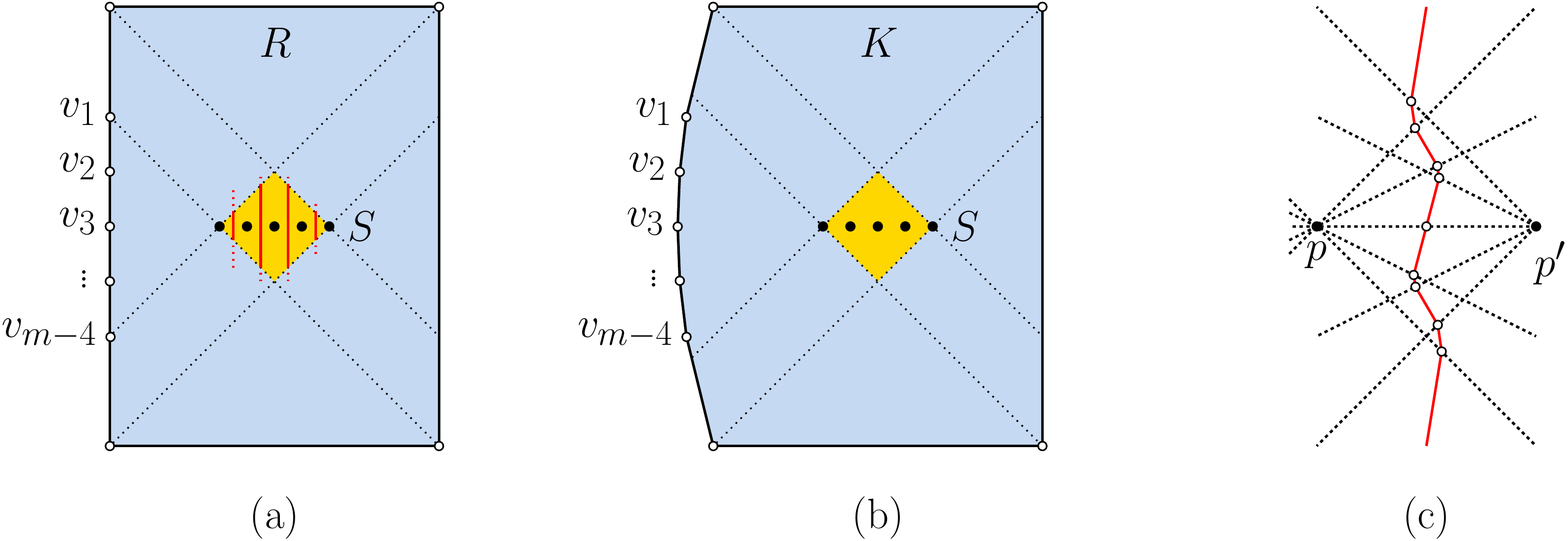}}
    \caption{Proof of Lemma~\ref{lem:lower-bound}.}\label{fig:lower-bound}
\end{figure}

To form $K$, we bend the left side of the rectangle out infinitesimally (see Figure~\ref{fig:lower-bound}(b)) so the points of $V$ become vertices of a convex polygon. The bending does not alter the shape of the Voronoi edges significantly, except to break each edge up into multiple arcs.

We assert that each consecutive pair of points $p, p' \in S$ contributes at least $m-4$ arcs to the Voronoi diagram. To see why, observe that there are $m-4$ sectors about $p$ and $p'$, and vertical line passing between $p$ and $p'$ intersects the boundaries between these sectors (see the red curve in Figure~\ref{fig:lower-bound}(c)). Due to the infinitesimal bending of $R$'s left side, each crossing produces a new segment on the bisector. Since there are $n-1$ consecutive pairs of sites, the total complexity of $\Vor_K(S)$ is at least $(m-4)(n-1) = \Omega(m n)$, as desired.

\end{proof}

\subsection{Bisector Segments} \label{sec:segment}

Next, let us consider the exact formulation of each bisector segment. Again, let $p$ and $p'$ be sites, let $z$ denote some point on the $(p,p')$-bisector, and let $e$ and $f$ (resp., $e'$ and $f'$) be the edges of $K$ bounding $\chi(z,p)$ (resp., $\chi(z,p')$). Let us label these so they encounter the chord in the order $\ang{e, z, p, f}$ (resp., $\ang{e', z, p', f'}$). Describing the bisector is messy, given that it depends on the coordinates/coefficients of $p$, $p'$, $e$, $f$, $e'$, and $f'$, for a total of $12$ parameters! To simplify matters, we reduce the problem to a simpler 1-dimensional problem. 

Consider the point $a$ where the extensions of the edges $e$ and $f$ meet, and define $a'$ analogously for $e'$ and $f'$ (see Figure~\ref{fig:segment-bisector}(a)). If $a = a'$ (implying that $e = e'$ and $f = f'$), then it is easy to see that the bisector degenerates to a straight line passing through $a$. Henceforth, let us assume that $a \neq a'$. 

\begin{figure}[htbp]
    \centerline{\includegraphics[scale=0.35]{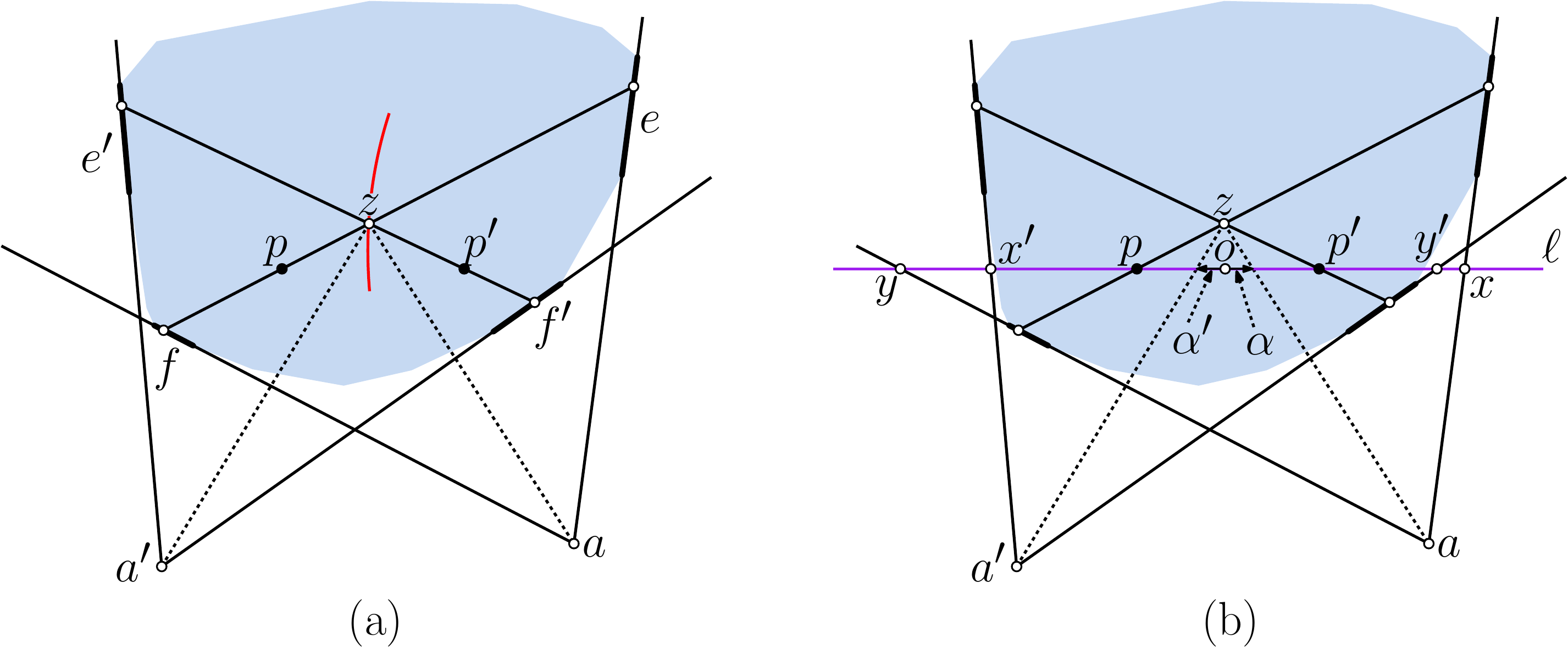}}
    \caption{Parameterizing a segment of the $(p,p')$-bisector.} \label{fig:segment-bisector}
\end{figure}

Because $z$ lies on the bisector, the cross ratio of $z$ and $p$ with respect to its chord equals the cross ratio of $z'$ and $p'$ with respect to its chord. Consider the line $\ell$ passing through $p$ and $p'$. Let $x$ and $y$ be the points lying on the intersection of $\ell$ where the extension of the edge $e$ and $f$ intersects $\ell$, respectively. Define $x'$, and $y'$ analogously for the edges $e'$ and $f'$ (see Figure~\ref{fig:segment-bisector}(b)). (The figure suggests that $p$ and $p'$ lie between the pairs $(x,y)$ and $(x',y')$, but this generally need not be the case. Irrespective of the order, the projection preserves cross ratios and hence preserves Hilbert distances.) Let $o$ be the point along $\ell$ such that $(o,p; y,x) = (o,p'; y',x')$. By consideration of Hilbert distances and the intermediate value theorem, it follows that such a point exists within the line segment $p p'$. Its uniqueness follows from simple algebraic properties of cross ratios.

Now that all the points are on a common line, we can think of these quantities as scalars. Define $\alpha$ and $\alpha'$ such that $o+\alpha$ is the point along $\ell$ where the line extending $z a$ intersects $\ell$, and $o+\alpha'$ is the point where the line extending $z a'$ intersects $\ell$. By basic projective geometry and the fact that $z$ is on the bisector, we have $(o+\alpha,p; y,x) = (o+\alpha',p'; y',x')$. Observe that any point on the bisector can be described by a pair $(\alpha,\alpha')$, by computing the associated points on $\ell$, shooting rays from $a$ and $a'$ through these respective points, and computing their intersection point. The following lemma shows that the reciprocals of $\alpha$ and $\alpha'$ have a simple linear relationship. The proof is rather technical, and appears in the appendix.

\begin{restatable}{lemma}{LemSB}
\label{lem:segment-bisector}
Given a convex polygon $K \in \RE^2$, sites $p, p' \in \interior(K)$, and a segment of the $(p,p')$-bisector determined by edges $e$, $f$, $e'$ and $f'$ as described above. For any point $z$ on the bisector, let $o+\alpha$ and $o+\alpha'$ denote the respective intersection points of the lines $z a$ and $z a'$ with the line $\ell$ through $p p'$. Then $1/\alpha$ and $1/\alpha'$ are related by a linear function depending on the above quantities.
\end{restatable}

This representation allows us to perform the essential operation needed by our Voronoi diagram algorithms of tracing a bisector segment. We can use one of the values $\alpha$ or $\alpha'$ as the principal parameter of the curve, and then use the linear relation to ascertain the value of the other parameter. We can then solve for the events of consequence, such as the end of the segment (when one of the chords through the bisector point encounters a new edge of $K$) or the location of a Voronoi vertex, where two bisectors meet. 

\section{Randomized Incremental Algorithm} \label{sec:randomized}

In this section we present a randomized incremental algorithm to compute the Hilbert Voronoi diagram for a set of $n$ sites contained in a $m$-sided convex polygon $K$. Our algorithm follows the structure of other incremental algorithms for Voronoi diagrams (see, e.g.,~\cite{guibas1992randomized,klein1993randomized}), and has the added benefit of generating a point-location data structure for the final diagram. This data structure locates the closest site to a point $q \in K$ in expected query time $O((\log n) (\log m n))$ (where the expectation is over the random choices made in the construction).

The sites are randomly permuted, and then inserted one by one into the diagram. The basis case is trivial, since the Voronoi cell of a single site is the entire body. To facilitate point location, we augment the diagram through the addition of line segment \emph{spokes} about each site. Recalling that each Voronoi cell is star shaped, about each site we add line segments connecting it to the boundary vertices of this cell, including its Voronoi vertices, the joints between consecutive bisector segments along each Voronoi edge, and the intersection of a Voronoi edge and the boundary of $K$, and the vertices of $K$ that lie within the Voronoi cell. This produces a (topological) triangulation of the diagram whose size is proportional to the size of the diagram itself. It may be stored in any standard data structure for planar subdivisions, like a doubly connected edge list~\cite{deberg2010book}. 

Locating a point in the current diagram is performed through the use of a \emph{history DAG}, that is, a directed acyclic graph whose terminal nodes correspond to triangles in the current diagram. The structure encodes the history of changes to the diagram. The addition of a new site will result in the creation of a number of new triangles, and the destruction of old ones. Each old triangle will store a reference to one of the new triangles that replaced it (as described in greater detail below).

The insertion process works as follows. When a new site $p_i$ is inserted, we first employ a search through the history DAG to determine the triangle of the current diagram containing it (see Figure~\ref{fig:insertion}(a)). From this triangle, we know the closest site in the current diagram, say $p_j$. We find the midpoint (in the Hilbert sense) on the line segment between $p_i$ and $p_j$. Starting at this point, we proceed to trace the boundary of the Voronoi cell of $p_i$ in the new diagram in counterclockwise order about $p_i$ (see Figure~\ref{fig:insertion}(b)). There are two cases that arise:
\begin{description}
\item[Bisector Trace:] We are tracing the bisector between $p_i$ and some existing site $p_k$ (see Figure~\ref{fig:insertion}(c)). We consider the sectors of $p_i$ and $p_k$ containing the current portion of the $(p_i,p_k)$-bisector. By applying the parameterization from Lemma~\ref{lem:segment-bisector}, we can trace the bisector until (1) we encounter the boundary of either sector, (2) we encounter the bisector between $p_k$ and another site, or (3) we encounter the boundary of $K$. 

In the first case, we create a new segment vertex here, add spokes to this vertex, erase the extension of the sector edge (shown as a broken line in Figure~\ref{fig:insertion}(c)) and continue the tracing in the new sector. In the second case, we create a new Voronoi vertex, add spokes to this vertex from its three defining sites, and continue the trace along the new bisector. In the third case, we insert two spokes joining the point where the bisector encounters the boundary to $p_i$ and $p_k$, respectively. We then transition to the boundary trace described next.

\item[Boundary Trace:] We are tracing the Voronoi cell of $p_i$ along $K$'s boundary. We walk along the boundary of $K$ in counterclockwise order, considering each consecutive sector of the closest site $p_k$, prior to $p_i$'s insertion. By applying the parameterization from Lemma~\ref{lem:segment-bisector}, we determine whether the $(p_i,p_k)$-bisector intersects $K$'s boundary within the intersection of the current pair of sectors. If so, we identify this point on the $\bd K$, add spokes to each of $p_i$ and $p_k$, and then resume tracing along the $(p_i,p_k)$-bisector. Otherwise, we encounter one of the two sector boundaries, and we continue the tracing the next sector.
\end{description} 

 Along the way, we erase spokes from the current diagram, and/or introduce new spokes connected to $p_i$. When we return to our starting point, the insertion is completed.

\begin{figure}[htbp]
    \centerline{\includegraphics[scale=0.35]{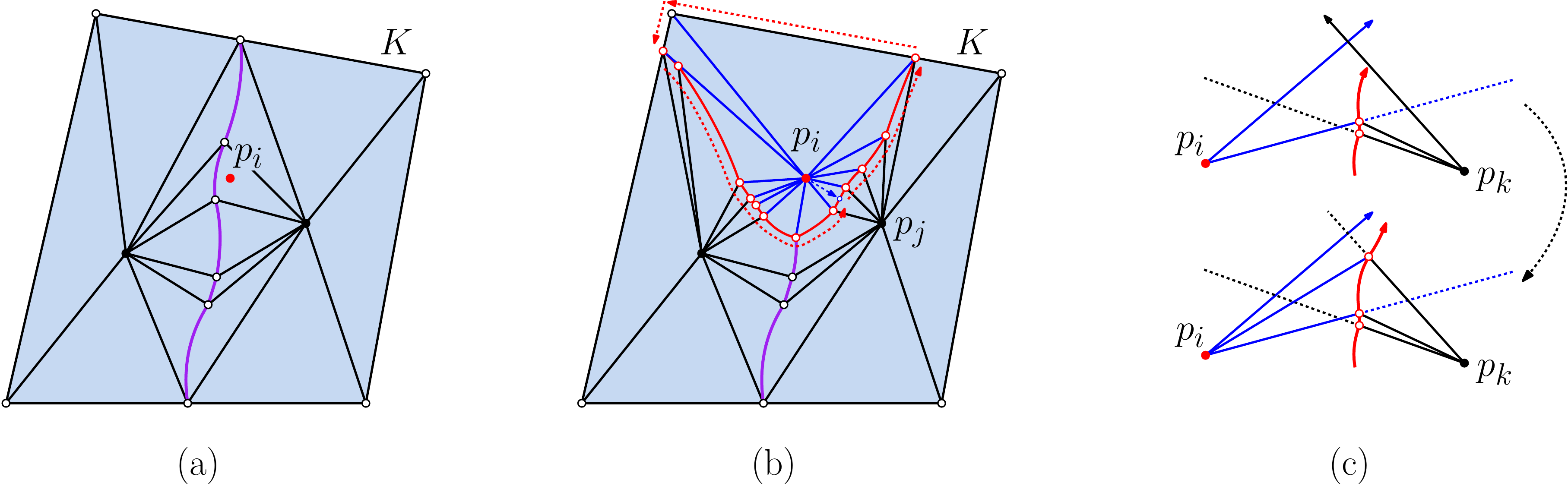}}
    \caption{Inserting a new site into the augmented diagram.} \label{fig:insertion}
\end{figure}

Point location in the history DAG is performed as follows for a query point $q$. The root node is associated with the entire polygon $K$. Each time a triangle is destroyed by the insertion of some point, we store a link to the site whose insertion caused the destruction (the root points to the first site inserted). We proceed as in Guibas, Knuth, and Sharir to determine the actual triangle containing the point in $O(\log m n)$ time as follows~\cite{guibas1992randomized}. Let $p_i$ denote the inserted site that caused $q$'s triangle to change, and let $p_j$ denote the site containing $q$ prior this event. Either $p_j$ is still $q$'s closest site or it is changed to $p_i$. We first compare the distance from $q$ to each of these sites. Once we have done this, we perform a radial binary search around this site to determine the triangle containing $q$. This radial search can be done in $O(\log m n)$ time. (As we shall see below, while the expected number of structural changes per insertion is $O(1)$ in expectation, we cannot infer that this is the case for an arbitrary query point. Also, we need to determine the triangle containing $q$ in order to bound the number of length of the search path for $q$ in the history DAG.) This completes the description of the algorithm. What remains is to analyze the algorithm's expected case running time.

\begin{lemma} \label{lem:exp-changes}
Assuming that sites are inserted in random order, the expected number of structural changes to the diagram over the course of the algorithm is $O(m n)$.
\end{lemma}

\begin{proof}
We employ a standard backwards analysis. Consider the final diagram consisting of $n$ sites. Because the sites have been inserted in random order, each site is equally likely to have been the last to be inserted. By Lemma~\ref{lem:total-complexity}, the average number of Voronoi edges on any Voronoi cell is $O(1)$ and the average number of bisector segments on each Voronoi edge is $O(m)$. 

The number of structural changes incurred by the insertion of a site is proportional to its number of Voronoi edges in the final diagram times the number of bisector segments on each of these edges (plus at most $m$ in case the Voronoi cell contains portions of the boundary of $K$). Therefore, the expected number of structural changes due the last insertion is $O(m)$. Since this holds irrespective of $n$, it follows that the total number of structural changes over all $n$ insertions is $O(m n)$.
\end{proof}

Next, we consider expected time needed to perform point location.

\begin{lemma} \label{lem:exp-search}
Let $q$ be any fixed point of the plane. If the point-location structure is constructed by inserting the sites in random order, then the expected cost of searching to locate the radial triangle containing $q$ in the final Voronoi diagram is $O((\log n)(\log m n))$.
\end{lemma}

\begin{proof}
The triangle containing $q$ depends on at most four sites. (Each triangle depends on the site $p_1$ where its two spokes meet. If the base of the triangle is a bisector, it also depends on the site $p_2$ on the other side of the bisector. Finally, either or both of remaining two vertices may be a Voronoi vertices, between $p_1$, $p_2$, and two other sites $p_3$ and $p_4$, respectively.) So, the probability that $q$ changes triangle membership during $i$th insertion is at most $\frac{4}{i}$. The total number of times $q$ descends to a new triangle in the history DAG in expectation is $\sum_{i=1}^n \frac{4}{i} = O(\log n)$. Each time such a descent occurs, we perform a radial search (described above) taking time $O(\log m n)$. Therefore, the total expected search time is $O((\log n) (\log m n))$.
\end{proof}

Combining the total expected update time and the expected search time for each of the $n$ sites, we have:

\begin{theorem} \label{thm:rand-inc-time}
Given an $m$-sided convex polygon $K$ in $\RE^2$ and a set of $n$ sites $S$ in $K$, the randomized incremental algorithm computes $\Vor_K(S)$ in expected time $O(n m + n(\log n) (\log m n))$ (where the expectation is over all possible insertion orders).
\end{theorem}

\section{Divide and Conquer Algorithm}

In this section we present an algorithm for constructing the Hilbert Voronoi diagram for a set of $n$ sites contained in a $m$-sided convex polygon $K$. Our algorithm follows the structure of the standard divide and conquer algorithm for the Euclidean Voronoi diagram (see, e.g.,~\cite{preparata1985book}).

The algorithm begins by partitioning the sites about a vertical line $\ell$ into two subsets, $S_L$ and $S_R$, of roughly equal sizes that lie to the left and right of this line, respectively. We recursively compute the Hilbert Voronoi diagrams $\Vor_K(S_L)$ and $\Vor_K(S_R)$ (see Figure~\ref{fig:left-right-bisector}(b) and~(c)). To merge them, we compute the bisector between $S_L$ and $S_R$, which we denote by $B(S_L, S_R)$. This consists of those points that are equidistant (in the Hilbert sense) from their closest sites in $S_L$ and $S_R$. We shall see below that this consists of a collection of curves, called \emph{components}, each being the concatenation of Voronoi edges between two sites, one in $S_L$ and one in $S_R$ (see Figure~\ref{fig:left-right-bisector}(d)). Each component terminates on the boundary $K$ (that is, no component is a closed loop). We will show below that it is possible to construct all the components of $B(S_L,S_R)$ by a single bottom-up traversal, which runs in $O(m n)$ time. Once the bisector is constructed, we complete the merging process, by discarding the portions of $\Vor_K(S_L)$ that lie on the right side of the bisector and the portions of $\Vor_K(S_R)$ that lie on the left side of bisector. Before describing the bisector-construction process, we explain the bisector's structure.

\begin{figure}[htbp]
    \centerline{\includegraphics[scale=0.35]{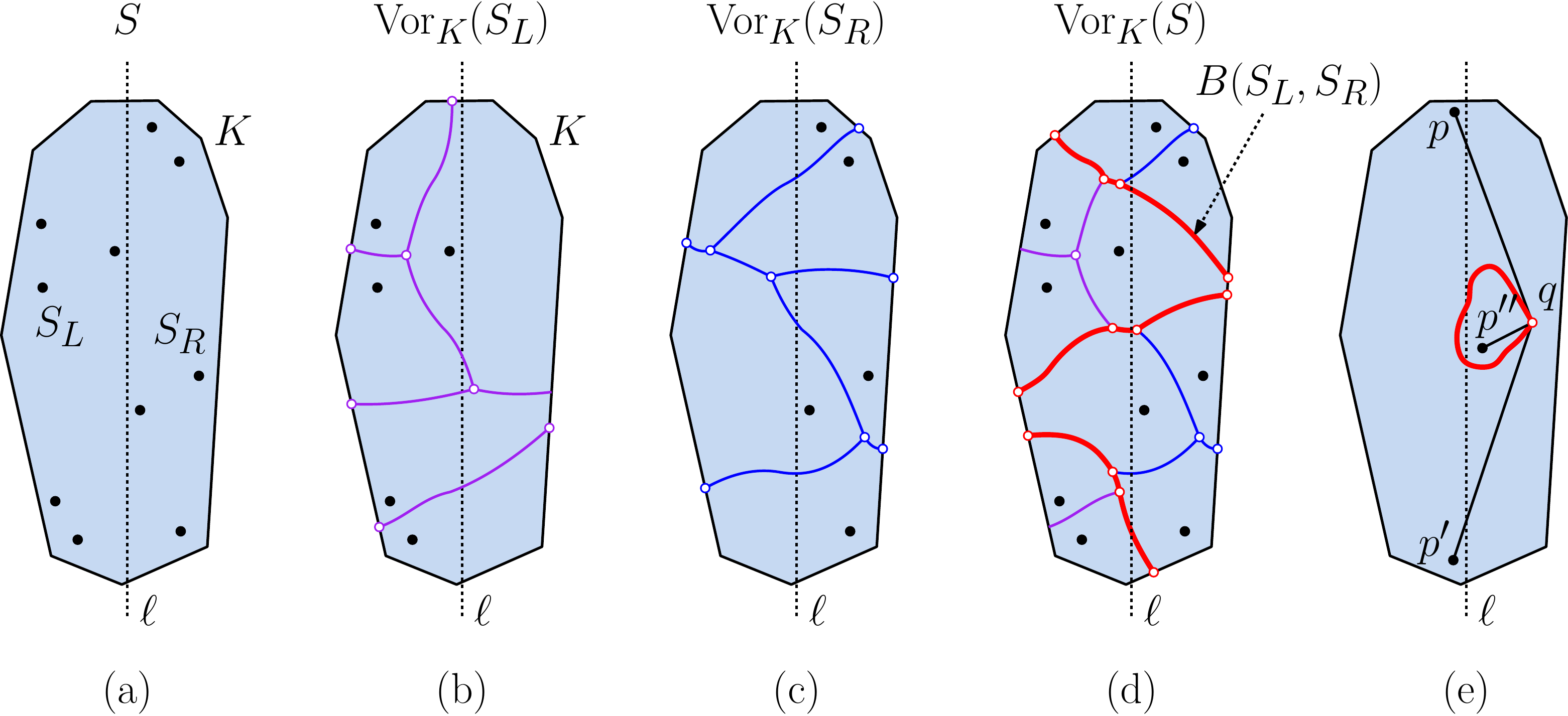}}
    \caption{Divide-and-conquer algorithm, bisectors, and the proof of Lemma~\ref{lem:left-right-bisector}(iii).} \label{fig:left-right-bisector}
\end{figure}

\begin{lemma} \label{lem:left-right-bisector}
Given two sets of sites $S_L$ and $S_R$ separated by a vertical line $\ell$, the Hilbert bisector $B(S_L, S_R)$ consists of a collection of pairwise disjoint curves, called \emph{components}, where:
\begin{enumerate}
\item[$(i)$] Each component is the concatenation of Voronoi edges from the final diagram, each between a site in $S_L$ and another in $S_R$.

\item[$(ii)$] Every point on $B(S_L,S_R)$ is visible to some point on $\ell \cap K$, that is, for each $x \in B(S_L,S_R)$ there exists a point $y \in \ell \cap K$ such that the open line segment $x y$ does not intersect $B(S_L,S_R)$.

\item[$(iii)$] Each component is terminated by endpoints that lie on the boundary of $K$ (and hence no component is a closed loop).

\item[$(iv)$] There are at most $n$ components, and the total combinatorial complexity of $B(S_L,S_R)$ is $O(m n)$.
\end{enumerate}
\end{lemma}

\begin{proof}
Claim~(i) follows from basic facts about Voronoi diagrams. Every point $q \in B(S_L,S_R)$ is on the Voronoi cell of at least two sites, one on each side of $\ell$. By Lemma~\ref{lem:star-shaped}, the line segment from $q$ to each site lies entirely within the Voronoi cell, and hence does not intersect $B(S_L,S_R)$. One of these two lines must cross $\ell \cap K$, which establishes~(ii).

To prove~(iii), suppose to the contrary that a component does not contact $K$'s boundary, implying that it forms a closed loop. Let $q$ be the point of this loop that is farthest from $\ell$ (that is, has the largest $x$-coordinate). By claim~(ii) and consideration points infinitesimally close on either side of $q$ on the loop, $q$ must be visible to two sites $p, p' \in S_L$ such that the segment $q p$ passes above the loop and $q p'$ passes below the loop (see Figure~\ref{fig:left-right-bisector}(e)). Therefore, $q$ is a Voronoi vertex between $p$, $p'$, and the closest site $p''$ in $S_R$. Thus, $q$ is the center of a Hilbert ball passing through these three sites. However, because the line segment $p p'$ lies entirely to the left of $\ell$ and $p''$ lies between these points to the right of $\ell$, this violates the convexity of Hilbert balls, a contradiction.

Claim~(iv) follows from the fact that each component contains at least one full Voronoi edge from the final diagram, and by Lemma~\ref{lem:total-complexity} there are at most $O(n)$ Voronoi edges in the final diagram. Also, the bisector is part of the final diagram, and hence its total complexity cannot exceed that of the final diagram, which by Lemma~\ref{lem:total-complexity} is $O(m n)$.
\end{proof}

Let us consider how the bisector $B(S_L,S_R)$ is constructed. As in the previous section, we will triangulate each Voronoi cell by adding \emph{spokes}, that is, line segments between each site and the various vertices (Voronoi vertices, bisector segment endpoints, vertices of $K$) that lie on the boundary of its Voronoi cell (see Figure~\ref{fig:left-right-trace}(a)). As mentioned above, we construct the bisector from bottom to top. (Note, however, that the components of $B(S_L,S_R)$ need not be vertically monotone.) We start at the lowest endpoint of $q$ of $\ell \cap K$. Let $p$ and $p'$ be the sites of $S_L$ and $S_R$ that are closest to $q$, respectively. We may assume that the construction algorithms for $\Vor_K(S_L)$ and $\Vor_K(S_R)$ provide us with the identities of $p$ and $p'$ as well as the triangles of their respective Voronoi cells that contain $q$. In $O(1)$ time, we can determine the closer of $p$ or $p'$ to $q$. If it is $p$, we begin a boundary trace counterclockwise about $\bd K$ and if it is $p'$, we begin a boundary trace clockwise about $\bd K$. These boundary traces are described below.

\begin{figure}[htbp]
    \centerline{\includegraphics[scale=0.35]{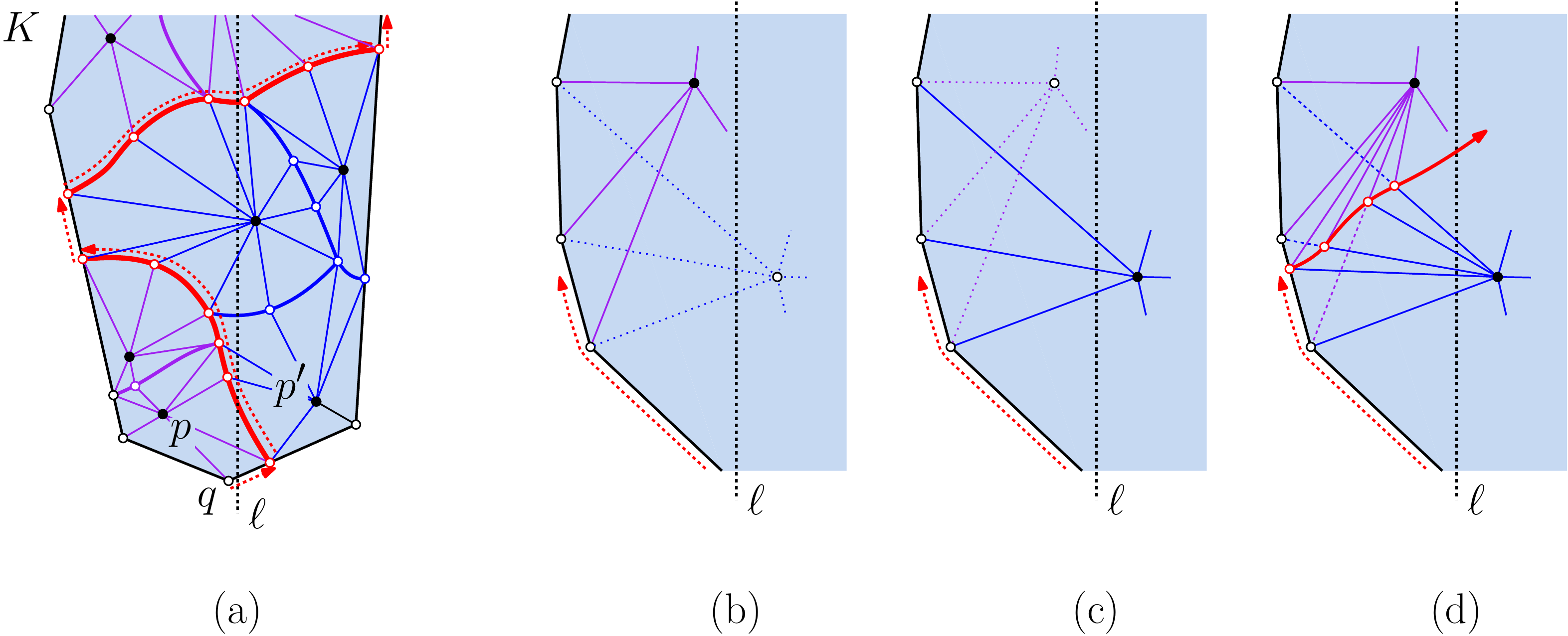}}
    \caption{Tracing $B(S_L, S_R)$. (Drawing is not geometrically accurate.)} \label{fig:left-right-trace}
\end{figure}

Generally there are two types tracings that we need to perform, depending on whether we are tracing a component of $B(S_L,S_R)$ or walking along $K$'s boundary in search of the next component. These are similar to the tracings we described in the case of the randomized algorithm. Note that generally, we are tracing counterclockwise about the sites of $S_L$ and clockwise about the sites of $S_R$ (see Fig.~\ref{fig:left-right-trace}).
\begin{description}
\item[Bisector Trace:] We are tracing the bisector between two sites $p_i \in S_L$ and $p_k \in S_R$ (see Figure~\ref{fig:insertion}(c)). We consider the sectors of $p_i$ and $p_k$ containing the current portion of the $(p_i,p_k)$-bisector. As in the bisector tracing for the randomized algorithm, we employ Lemma~\ref{lem:segment-bisector} to trace the bisector until either (1) it intersects one of the two sector edges, (2) hits another bisector, or (3) hits $K$'s boundary. 

In the first case, we erase the extension of the sector edge (shown as a broken line in Figure~\ref{fig:left-right-trace}(d)), create a new segment vertex here, add spokes to this vertex, and continue the tracing in the new sector. In the second case, we create a new Voronoi vertex, add spokes to this vertex from its three defining sites, and continue the trace along the appropriate bisector (always walking on a bisector between a left-side site and a right-side site). In the third case, we insert two spokes joining the point where the bisector encounters the boundary to $p_i$ and $p_k$, respectively. We then transition to the boundary trace described next.

\item[Boundary Trace:] There are two cases, depending on whether we are tracing along the boundary of a cell of $S_L$ or $S_R$. In the former case, we walk counterclockwise along $K$'s boundary until reaching the end of the current sector, and in the latter we walk clockwise. We consider each consecutive sector of the closest sites $p_i$ from $S_L$ and $p_j$ from $S_R$. By applying the parameterization from Lemma~\ref{lem:segment-bisector}, we determine whether the $(p_i,p_k)$-bisector intersects $K$'s boundary within the intersection of the current pair of sectors. If so, we identify this point on the $\bd K$, add spokes to each of $p_i$ and $p_k$, and then resume tracing along the $(p_i,p_k)$-bisector (see Figure~\ref{fig:left-right-trace}(b)--(d)). Otherwise, we encounter one of the two sector boundaries, and we continue the tracing along the boundary.
\end{description} 

The tracing process runs in $O(1)$ time for each newly added segment to the bisector, for a total of $O(m n)$ time. The overhead for the divide-and-conquer is dominated by the tracing time, which implies that divide-and-conquer algorithm's overall running time satisfies the recurrence $T(n) = 2 T(n/2) + m n$. This solves to $O(n m \log n)$. In summary, we obtain the following result.

\begin{theorem} \label{thm:DandC-time}
Given an $m$-sided convex polygon $K$ in $\RE^2$ and a set of $n$ sites $S$ in $K$, the divide-and-conquer algorithm computes $\Vor_K(S)$ in time $O(n m \log n)$.
\end{theorem}

\appendix
\section{Appendix} \label{sec:appendix}

{\LemSB*}

\begin{proof}
To simplify notation, let $\underline{x y}$ be a shorthand for $x - y$. Given that $(o+\alpha,p; x,y) = (o+\alpha',p'; x',y')$, we have
\[
    \frac{\underline{x (o+\alpha)} \cdot \underline{p y}}{\underline{x p} \cdot \underline{(o+\alpha) y}}
        ~ = ~ \frac{\underline{x' (o+\alpha')} \cdot \underline{p' y'}}{\underline{x' p'} \cdot \underline{(o+\alpha') y'}}.
\]
Observe that $\underline{(x+y)z} = x+y-z =\underline{x z} + y$ and $\underline{z(x+y)} = z-x-y = \underline{z x} - y$. Expanding both sides, yields
\begin{align*}
\MoveEqLeft
    \left( \underline{x (o+\alpha)} \cdot \underline{p y} \right) \left( \underline{x' p'} \cdot \underline{(o+\alpha') y'} \right)
     ~ = ~ \left( \underline{x' (o+\alpha')} \cdot \underline{p' y'} \right) \left( \underline{x p} \cdot \underline{(o+\alpha) y} \right) \\
\MoveEqLeft
    \left( (\underline{x o} - \alpha) \cdot \underline{p y} \right) \left( \underline{x' p'} \cdot (\underline{o y'} +\alpha')\right)
     ~ = ~ \left( (\underline{x' o} - \alpha') \cdot \underline{p' y'} \right) \left( \underline{x p} \cdot (\underline{o y} + \alpha) \right) \\
\MoveEqLeft
    \big( \underline{x o} \cdot \underline{p y} - \alpha \cdot \underline{p y} \big) 
    \big( \underline{x' p'} \cdot \underline{o y'} + \underline{x' p'} \cdot \alpha' \big)
        ~ = ~ \big( \underline{x' o} \cdot \underline{p' y'} - \alpha' \cdot \underline{p' y'} \big) 
                \big( \underline{x p} \cdot \underline{o y} + \underline{x p} \cdot \alpha \big) \\
\MoveEqLeft
    \underline{x o} \cdot \underline{p y} \cdot \underline{x' p'} \cdot \underline{o y'} 
    - \alpha \cdot \underline{p y} \cdot \underline{x' p'} \cdot \underline{o y'}
    + \underline{x o} \cdot \underline{p y} \cdot \underline{x' p'} \cdot \alpha' 
    - \alpha \cdot \underline{p y} \cdot \underline{x' p'} \cdot \alpha' \\
        & ~ = ~ 
    \underline{x' o} \cdot \underline{p' y'} \cdot \underline{x p} \cdot \underline{o y} 
    - \alpha' \cdot \underline{p' y'} \cdot \underline{x p} \cdot \underline{o y} 
    + \underline{x' o} \cdot \underline{p' y'} \cdot \underline{x p} \cdot \alpha 
    - \alpha' \cdot \underline{p' y'} \cdot \underline{x p} \cdot \alpha 
\end{align*}
Collecting common factors of $\alpha$ and $\alpha'$, we have
\begin{align}\label{eq:seg-bis-a}
\MoveEqLeft
    \left( \underline{x o} \cdot \underline{p y} \cdot \underline{x' p'} \cdot \underline{o y'} 
    - \underline{x' o} \cdot \underline{p' y'} \cdot \underline{x p} \cdot \underline{o y} \right)
    - \left( \underline{p y} \cdot \underline{x' p'}
    - \underline{p' y'} \cdot \underline{x p} \right) \alpha \alpha'  \nonumber \\ 
        & ~ = ~ 
    \left( \underline{p y} \cdot \underline{x' p'} \cdot \underline{o y'} 
    + \underline{x' o} \cdot \underline{p' y'} \cdot \underline{x p} \right) \alpha 
    - \left( \underline{x o} \cdot \underline{p y} \cdot \underline{x' p'} 
    + \underline{p' y'} \cdot \underline{x p} \cdot \underline{o y} \right) \alpha'.
\end{align}
By definition of $o$, we have 
\[
    \frac{\underline{x o} \cdot \underline{p y}}{\underline{x p} \cdot \underline{o y}}
        ~ = ~ \frac{\underline{x' o} \cdot \underline{p' y'}}{\underline{x' p'} \cdot \underline{o y'}} 
    \quad\Longrightarrow\quad
    \underline{x o} \cdot \underline{p y} \cdot \underline{x' p'} \cdot \underline{o y'} -
        \underline{x' o} \cdot \underline{p' y'} \cdot \underline{x p} \cdot \underline{o y} ~ = ~ 0.
\]
Thus, the first term in Eq.\eqref{eq:seg-bis-a} vanishes, which yields
\[
    \left( \underline{p y} \cdot \underline{x' p'}
    - \underline{p' y'} \cdot \underline{x p} \right) \alpha \alpha'
        ~ = ~ 
    \left( \underline{x o} \cdot \underline{p y} \cdot \underline{x' p'} 
    + \underline{p' y'} \cdot \underline{x p} \cdot \underline{o y} \right) \alpha'
    - \left( \underline{p y} \cdot \underline{x' p'} \cdot \underline{o y'} 
    + \underline{x' o} \cdot \underline{p' y'} \cdot \underline{x p} \right) \alpha.
\]
Letting:
\begin{align*}
    A & ~ = ~ \underline{p y} \cdot \underline{x' p'} - \underline{p' y'} \cdot \underline{x p}, \quad
    B   ~ = ~ \underline{x o} \cdot \underline{p y} \cdot \underline{x' p'} + \underline{p' y'} \cdot \underline{x p} \cdot \underline{o y},
    \quad \mathrm{and} \\
    C & ~ = ~ \underline{p y} \cdot \underline{x' p'} \cdot \underline{o y'} + \underline{x' o} \cdot \underline{p' y'} \cdot \underline{x p}, 
\end{align*}
and dividing by $\alpha \alpha'$, we see that the bisector is given by the following linear relation in $1/\alpha$ and $1/\alpha'$
\[
     A ~ = ~ \frac{B}{\alpha} - \frac{C}{\alpha'},
\]
as desired.
\end{proof}

\bibliography{shortcuts,hilbert}

\begin{thebibliography}{10}

\bibitem{AAFM19}
Ahmed Abdelkader, Sunil Arya, Guilherme~Dias da~Fonseca, and David~M. Mount.
\newblock Approximate nearest neighbor searching with non-euclidean and
  weighted distances.
\newblock In {\em Proc.\ 30th Annu.\ ACM-SIAM Sympos.\ Discrete Algorithms},
  pages 355--372, 2019.
\newblock \href {https://doi.org/10.1137/1.9781611975482.23}
  {\path{doi:10.1137/1.9781611975482.23}}.

\bibitem{AbM18}
Ahmed Abdelkader and David~M. Mount.
\newblock Economical {Delone} sets for approximating convex bodies.
\newblock In {\em Proc.\ 16th Scand.\ Workshop Algorithm Theory}, pages
  4:1--4:12, 2018.

\bibitem{AAFM20}
Rahul Arya, Sunil Arya, Guilherme~Dias da~Fonseca, and David~M. Mount.
\newblock Optimal bound on the combinatorial complexity of approximating
  polytopes.
\newblock In {\em Proc.\ 31st Annu.\ ACM-SIAM Sympos.\ Discrete Algorithms},
  pages 786--805, 2020.
\newblock URL: \url{https://arxiv.org/abs/1910.14459}, \href
  {https://doi.org/10.1137/1.9781611975994.48}
  {\path{doi:10.1137/1.9781611975994.48}}.

\bibitem{AFM17b}
Sunil Arya, Guilherme~Dias da~Fonseca, and David~M. Mount.
\newblock Near-optimal $\varepsilon$-kernel construction and related problems.
\newblock In {\em Proc.\ 33rd Internat.\ Sympos.\ Comput.\ Geom.}, pages
  10:1--15, 2017.
\newblock URL: \url{https://arxiv.org/abs/1604.01175}.

\bibitem{AFM17c}
Sunil Arya, Guilherme~Dias da~Fonseca, and David~M. Mount.
\newblock On the combinatorial complexity of approximating polytopes.
\newblock {\em Discrete Comput.\ Geom.}, 58(4):849--870, 2017.
\newblock \href {https://doi.org/10.1007/s00454-016-9856-5}
  {\path{doi:10.1007/s00454-016-9856-5}}.

\bibitem{AFM17a}
Sunil Arya, Guilherme~Dias da~Fonseca, and David~M. Mount.
\newblock Optimal approximate polytope membership.
\newblock In {\em Proc.\ 28th Annu.\ ACM-SIAM Sympos.\ Discrete Algorithms},
  pages 270--288, 2017.

\bibitem{busemann1955geodesics}
Herbert Busemann.
\newblock {\em The Geometry of Geodesics}.
\newblock Academic Press, 1955.

\bibitem{deberg2010book}
Mark de~Berg, Otfried Cheong, Marc van Kreveld, and Mark Overmars.
\newblock {\em Computational Geometry: {A}lgorithms and Applications}.
\newblock Springer, 3rd edition, 2010.

\bibitem{EHN11}
Friedrich Eisenbrand, Nicolai H{\"a}hnle, and Martin Niemeier.
\newblock Covering cubes and the closest vector problem.
\newblock In {\em Proc.\ 27th Annu.\ Sympos.\ Comput.\ Geom.}, pages 417--423,
  2011.

\bibitem{EiV21}
Friedrich Eisenbrand and Moritz Venzin.
\newblock Approximate {CVPs} in time $2^{0.802 n}$.
\newblock {\em Journal of Computer and System Sciences}, 2021.

\bibitem{guibas1992randomized}
Leonidas~J. Guibas, Donald~E. Knuth, and Micha Sharir.
\newblock Randomized incremental construction of {Delaunay} and {Voronoi}
  diagrams.
\newblock {\em Algorithmica}, 7:381--413, 1992.
\newblock \href {https://doi.org/10.1007/BF01758770}
  {\path{doi:10.1007/BF01758770}}.

\bibitem{hilbert1895linie}
D.~Hilbert.
\newblock Ueber die gerade linie als k{\" u}rzeste verbindung zweier punkte.
\newblock {\em Mathematische Annalen}, 46:91--96, 1895.

\bibitem{klein1993randomized}
Rolf Klein, Kurt Mehlhorn, and Stefan Meiser.
\newblock Randomized incremental construction of abstract {Voronoi} diagrams.
\newblock {\em Comput.\ Geom.\ Theory Appl.}, 3(3):157--184, 1993.
\newblock URL:
  \url{https://www.sciencedirect.com/science/article/pii/0925772193900333},
  \href {https://doi.org/https://doi.org/10.1016/0925-7721(93)90033-3}
  {\path{doi:https://doi.org/10.1016/0925-7721(93)90033-3}}.

\bibitem{NaV19}
M{\'a}rton Nasz{\'o}di and Moritz Venzin.
\newblock Covering convex bodies and the closest vector problem.
\newblock {\em arXiv preprint arXiv:1908.08384}, 2019.

\bibitem{nielsen2017balls}
Frank Nielsen and Laetitia Shao.
\newblock On balls in a {Hilbert} polygonal geometry (multimedia contribution).
\newblock In {\em Proc.\ 33rd Internat.\ Sympos.\ Comput.\ Geom.}, volume~77 of
  {\em Leibniz International Proceedings in Informatics (LIPIcs)}, pages
  67:1--67:4. Schloss Dagstuhl--Leibniz-Zentrum f{\"u}r Informatik, 2017.
\newblock URL: \url{http://drops.dagstuhl.de/opus/volltexte/2017/7244}, \href
  {https://doi.org/10.4230/LIPIcs.SoCG.2017.67}
  {\path{doi:10.4230/LIPIcs.SoCG.2017.67}}.

\bibitem{papadopoulos2014funk}
Athanase Papadopoulos and Marc Troyanov.
\newblock From {Funk} to {Hilbert} geometry.
\newblock In {\em Handbook of {Hilbert} geometry}, volume~22 of {\em {IRMA}
  Lectures in Mathematics and Theoretical Physics}, pages 33--68. European
  Mathematical Society Publishing House, 2014.
\newblock URL: \url{https://hal.archives-ouvertes.fr/hal-01015584}, \href
  {https://doi.org/10.4171/147-1/2} {\path{doi:10.4171/147-1/2}}.

\bibitem{papadopoulos2014handbook}
Athanase Papadopoulos and Marc Troyanov.
\newblock {\em Handbook of {Hilbert} geometry}, volume~22 of {\em {IRMA}
  Lectures in Mathematics and Theoretical Physics}.
\newblock European Mathematical Society Publishing House, 2014.

\bibitem{preparata1985book}
Franco~P. Preparata and Michael~Ian Shamos.
\newblock {\em Computational Geometry: An Introduction}.
\newblock Springer, 1985.

\bibitem{RoV21}
Thomas Rothvoss and Moritz Venzin.
\newblock Approximate {CVP} in time $2^{0.802 n}$ -- now in any norm!
\newblock {\em arXiv preprint arXiv:2110.02387}, 2021.

\bibitem{troyanov2014funk}
Marc Troyanov.
\newblock {Funk} and {Hilbert} geometries from the {Finslerian} viewpoint.
\newblock In {\em Handbook of {Hilbert} geometry}, volume~22 of {\em {IRMA}
  Lectures in Mathematics and Theoretical Physics}, pages 69--110. European
  Mathematical Society Publishing House, 2014.
\newblock URL: \url{https://hal.archives-ouvertes.fr/hal-01015584}, \href
  {https://doi.org/10.4171/147-1/2} {\path{doi:10.4171/147-1/2}}.

\bibitem{yamada2014convex}
Sumio Yamada.
\newblock Convex bodies in {Euclidean} and {Weil-Petersson} geometries.
\newblock {\em Proc.\ Amer.\ Math.\ Soc.}, 142(2):603--616, 2014.
\newblock \href {https://doi.org/10.1090/S0002-9939-2013-11841-1}
  {\path{doi:10.1090/S0002-9939-2013-11841-1}}.

\end{thebibliography}

\end{document}